\documentclass{article} 
\usepackage{bm}
\usepackage{amsthm}

\usepackage{thmtools}
\usepackage{thm-restate}

\usepackage{hyperref,cleveref}
\usepackage{fullpage}
\usepackage{amsmath,amssymb}
\usepackage{fancybox}
\def\MM{{\mathbf{M}}}
\def\DD{{\mathbf{D}}}
\def\CC{{\mathbf{C}}}

\newdimen\pIR
\newcommand\R{{\rm I\kern\pIR R}}
\def\Reals#1{\mathbb{R}^{#1}}

\newcommand{\Half}{\frac{1}{2}}

\newtheorem{theorem}{Theorem}[section]
\newtheorem{definition}[theorem]{Definition}
\newtheorem{lemma}[theorem]{Lemma}
\newtheorem{fact}[theorem]{Fact}
\newtheorem{proposition}[theorem]{Proposition}

\newcommand{\M}[1]{\mathbf{#1}}
\newcommand{\SM}[1]{\mathbf{\bar{#1}}}

\newcommand{\V}[1]{\bm{#1}}
\newcommand{\T}{\top}

\newcommand{\MG}[2]{\mathcal{N}(#1,#2)}
\newcommand{\OO}[1]{O(#1)}

\newcommand{\Nsp}[1]{ {\| #1 \|}_{2}}

\newcommand{\Mij}[3]{#1_{{#2},{#3}}}

\title{Scalable Parallel Factorizations of SDD Matrices
  and Efficient Sampling for Gaussian Graphical Models}

\author{Dehua Cheng\footnotemark[1]\\USC \and Yu Cheng\footnotemark[2]\\USC \and Yan Liu\thanks{Supported in
  part by NSF research grants IIS-1134990, IIS-1254206 and U.S. Defense Advanced Research Projects Agency (DARPA) under Social Media in Strategic Communication (SMISC) program, Agreement Number W911NF-12-1-0034.}\\USC \and Richard Peng\\MIT \and Shang-Hua Teng\thanks{Supported in
  part by NSF grants CCF-1111270 and CCF-096448 and
  by the Simons Investigator Award from the Simons Foundation.}\\ USC}

\begin{document}

\maketitle

\begin{abstract}
Motivated by a sampling problem basic to computational statistical inference, 
   we develop a nearly optimal algorithm for a fundamental problem in
   spectral graph theory and numerical analysis.
  Given an $n\times n$ SDDM matrix ${\bf \mathbf{M}}$,
  and a constant $-1 \leq p \leq 1$, our algorithm gives efficient
  access to a sparse $n\times n$
  linear operator $\tilde{\mathbf{C}}$ such that 
$${\mathbf{M}}^{p} \approx \tilde{\mathbf{C}} \tilde{\mathbf{C}}^\top.$$
The solution is based on factoring ${\bf \mathbf{M}}$ into a product
of simple and sparse matrices using squaring and spectral sparsification.

For ${\mathbf{M}}$ with $m$ non-zero entries, 
our algorithm takes work nearly-linear in $m$,
and polylogarithmic depth on  a parallel machine with $m$ processors.
This gives the first sampling algorithm 
  that only requires {\em nearly linear work} and 
  \textit{$n$ i.i.d. random univariate Gaussian samples}
  to generate {\em i.i.d. random samples} for
  $n$-dimensional Gaussian random fields with SDDM precision matrices.
  For sampling this natural subclass of Gaussian random fields,
  it is optimal in the randomness and nearly optimal 
  in the work and parallel complexity.
 In addition, our sampling algorithm can be directly extended to Gaussian random fields
  with SDD precision matrices.

\end{abstract}

\section{Introduction}

Sampling from a multivariate probability distribution is one of the most
  fundamental problems in statistical
  machine learning and statistical inference.
In the sequential computation framework,
  the Markov Chain Monte Carlo (MCMC) based
  Gibbs sampling algorithm and its theoretical underpinning built on the
  celebrated Hammersley-Clifford Theorem have provided the algorithmic and
  mathematical foundation for analyzing graphical models~\cite{jordan98, Koller:2009:PGM:1795555}.
However, unless there are significant independences among the variables,
  the Gibbs sampling process tends to be highly sequential~\cite{jordan98}.
Moreover, many distributions require the Gibbs sampling process
  to take a significant number of iterations to
  achieve a reliable estimate.
Despite active efforts by various research groups~\cite{
  Gonzalez+al:aistatspgibbs, Ihler03efficientmultiscale, Johnson13,welling14},
  the design of a scalable parallel Gibbs sampling algorithm remains a challenging problem.
The ``holy grail'' question in parallel sampling can be characterized as:
\begin{quote}
{\em Is it possible to obtain a characterization of the family of Markov random
  fields from which one can draw a sample in polylogarithmic parallel time
  with nearly-linear total work?}
\end{quote}

Formally, the performance of a parallel sampling algorithm can be characterized
  using the following three parameters:
\begin{itemize}
\item {\bf Time complexity (work)}: 
  the total amount of operations needed 
  by the sampling algorithm to generate the first and subsequent
  random samples.  
\item {\bf Parallel complexity (depth)}:  
  the length of the longest critical path
  of the algorithm in generating the first and subsequent random samples.
\item {\bf Randomness complexity}:  total random numbers (scalars)
  needed by the sampling algorithm to generate a random sample.
\end{itemize}

One way to obtain an (approximately) random sample is through Gibbs sampling. 
The basic Gibbs sampling method follows a randomized
  process that iteratively
  resamples each variable conditioned on the value of other variables.
This method is analogous to the 
  Gauss-Seidel method for solving linear systems,
  and this connection is most apparent on an important class of multivariate probability distributions:
{\em multivariate normal distributions} or {\em multivariate Gaussian distributions} \cite{LohWainwright}.
In the framework of graphical models, such a distribution is commonly specified by 
  its precision matrix  $\M{\Lambda}$ and potential vector $\V{h} $, i.e.,
  $\text{Pr}(\mathbf{x}| \M{\Lambda}, \V{h}) \propto \exp (-\frac{1}{2} \mathbf{x}^{\T} \M{\Lambda} \mathbf{x} + \V{h}^{\T} \mathbf{x})$, which defines 
  a {\em Gaussian random field}.

The analogy between the Gibbs sampling method and Gauss-Seidel 
  iterative method for solving
  linear systems imply that it has similar worst-case behaviors.
For example, if a Gaussian random field has $n$ variables with a chain structure,
  it may take $\Omega(n)$ iterations for the Gibbs sampling algorithm to reasonably converge \cite{LiuKW13}.
Given an $n\times n$ precision matrix with $m$ non-zero entries, 
  the Gibbs sampling process may take $\Omega(mn)$  total work
  and require $\Omega(n \chi(G) )$ iterations, where $\chi(G)$ is
  the chromatic number of the graph underlying the precision matrix. 
It may also use $\Omega(n^2)$ random numbers 
  ($n$ per iteration) to generate its first sample. 
However, in practice it usually performs significantly better than these worst case bounds.
 
Our study is partly motivated by the recent work of Johnson,
  Saunderson, Willsky~\cite{Johnson13} that provides a mathematical
  characterization of the {\em Hogwild Gaussian Gibbs sampling} heuristics
  (see also Niu, Recht, R\'{e}, and Wright \cite{NiuHogwild}).
Johnson {\em et al.} proved that if the precision matrix $\M{\Lambda}$
  is {\em symmetric generalized diagonally dominant}, aka. an {\em H-matrix}
  \footnote{A symmetric real $n\times n$ matrix $\MM = (m_{i,j})$
    is generalized diagonally dominant if it can be scaled into
    a diagonally dominant matrix, that is, there exists a positive
    diagonal scaling matrix $\DD = {\rm diag}([d_1,...,d_n])$
    such that $\DD \MM \DD$ is strictly diagonally dominant.},
  then Hogwild Gibbs sampling converges with the correct mean
  for any variable partition (among parallel processors) and
  for any number of locally sequential Gibbs sampling steps
  with old statistics from other processors.
This connection to H-matrices leads us to draw from developments in
  nearly-linear work solvers for linear systems in symmetric diagonally dominant (SDD)
  matrices~\cite{SpielmanTengSolver,MillerPeng,Kelner,PengSpielman},
  which are algorithmically inter-reducible to H-matrices~\cite{DaitchS08}.
Most relevant to the parallel sampling question is a result with worst-case logarithmic
  upper bound on parallel complexity~\cite{PengSpielman}.
In this paper we develop a scalable parallel sampling algorithm 
  for Gaussian random fields with SDD precision matrices
  by extending the algorithmic framework of~\cite{PengSpielman}
  from solving linear systems to sampling graphical models.

\subsection{Our Contributions}

We focus on the problem of parallel sampling for
  a Gaussian random field
  whose precision matrix $\M{\Lambda}$ is {\em SDDM},
  which are positive definite matrices with non-positive off-diagonals.
This is a natural subclass of
  Gaussian random fields which arises in applications such as image
  denoising \cite{ImageDenoising}.
It can also be used in the rejection/importance sampling
  framework~\cite{andrieu2003introduction} as a proposed distribution,
  likely giving gains over the commonly used diagonal precision matrices in this framework.

Although faster sampling is possible for some subset of 
  SDDM precision matrices \cite{LiuKW13,PapandreouYuille},
  $\Omega(mn)$ remains a worst-case complexity barrier
  for generating unbiased samples in this 
  subclass of Gaussian random fields.
While Newton's method leads to sampling algorithms with
  polylogarithmic parallel steps, 
  these algorithms are expensive due to calls to dense matrix multiplication,
  {\em even if} the original graphical model is sparse. 
Our algorithm can generate i.i.d. samples 
  from the Gaussian random field with an approximate covariance
  with nearly-linear total work.
In contrast, samples generated from Gibbs sampling are correlated, 
  but are drawn from the exact covariance in the limit.

Our algorithm is based on the following three-step 
  numerical framework for sampling a Gaussian random field with parameters
  $\left(\M{\Lambda}, \V{h}\right)$: $\text{Pr}(\mathbf{x}| \M{\Lambda}, \V{h}) \propto \exp (-\frac{1}{2} \mathbf{x}^{\T} \M{\Lambda} \mathbf{x} + \V{h}^{\T} \mathbf{x})$.
\begin{enumerate}
\item Compute the mean $\V{\mu}$ of the distribution from
  $\left(\M{\Lambda}, \V{h}\right)$ by solving the linear system
  $\M{\Lambda}\V{\mu}=\V{h}$.

\item Compute a factor of the covariance matrix by finding 
  an inverse square-root factor of the precision matrix $\M{\Lambda}$, i.e.,
\begin{align}
  \M{\Lambda}^{-1}=\M{C}\M{C}^{\T}.
\end{align}

\item If $\M{C}$ is an $n\times n'$ matrix, then we can obtain a sample of 
  $\MG{\V{\mu}}{\M{\Lambda}}$ by first drawing a random vector
  $\V{z}={\left(z_1,z_2,\dots,z_{n'}\right)}^{\T}$,
  where $z_i,i=1,2,\dots,n'$ are i.i.d. standard Gaussian random variables,
  and return $\V{x} =\M{C}\V{z}+\V{\mu}$.
\end{enumerate}

If $\M{C}$ has an efficient sparse representation, then Step 3 has a scalable
  parallel implementation with parallel complexity polylogarithmic 
  in $n$.
Polylogarithmic depth and nearly-linear work solvers for SDDM linear systems
  \cite{PengSpielman} also imply a scalable 
  parallel solution to Step 1.
So this framework allows us to 
  concentrate on step 2: factoring the covariance matrix $\M{\Lambda}$.

We give two approaches that take nearly-linear work and polylogarithmic parallel depth
  to construct a linear operator $\tilde{\CC}$ 
  such that $\M{\Lambda}^{-1} \approx \tilde{\M{C}} \tilde{\M{C}}^{\T}$.
The first approach factorizes $\M{\Lambda}$ based on an
  analog of the \emph{signed edge-vertex incidence matrix}.
It differs from previous approaches such as the one by 
  Chow and Saad~\cite{chow14precon}
  in that it generates a non-square $\tilde{\M{C}}$.
This leads to a sampling algorithm with
  randomness complexity of $m$, the number of nonzeros in $\M{\Lambda}$: 
  it generates a sample of $(\M{\Lambda},\V{h})$
  by first taking random univariate Gaussian samples,
  one per non-zero entry in $\M{\Lambda}$.
This sampling approach can be effective if $m = O(n)$.

The second approach, our main result,  
  produces a \emph{square} $n\times n$ inverse
  factorization $\tilde{\CC}$.
It produces, after efficient preprocessing, a scalable parallel {\em i.i.d. random vector} generator for $(\M{\Lambda},\V{h})$
  with optimal randomness complexity: each random sample of 
  $(\M{\Lambda},\V{h})$ is generated from
  $n$ {\em i.i.d.} standard Gaussian random variables.
This factorization problem motivates us to develop a
  highly efficient algorithm 
  for a fundamental problem in {\em spectral graph theory} that could be significant in its own right:
\begin{quote}
Given an $n\times n$ SDDM matrix $\MM$,
  a constant $-1 \leq p \leq 1$, and an approximation parameter $\epsilon$,
  compute a sparse representation of
  an  $n\times n$  linear operator $\tilde{\CC}$ such that 
\begin{align}
  \M{M}^{p} \approx_{\epsilon} \tilde{\CC} \tilde{\CC}^\T,
\end{align}
\end{quote}
where $\approx_{\epsilon}$ is spectral similarity between linear operators
which we will define at the start of Section~\ref{sec:background}.
Our algorithm factors $\M{M}^{p}$ into a series of well-conditioned,
  easy-to-evaluate, sparse matrices, which we call a sparse factor chain. 
For any constant $-1 \leq p \leq 1$, our algorithm 
  constructs a sparse factor chain with 
  $O(m\cdot \log^{c_1} n\cdot \log^{c_2} (\kappa/\epsilon) \cdot \epsilon^{-4})$
  work for modest
  \footnote{Currently, both constants are at most $5$ -- these constants
    and the exponent $4$ for $(1/\epsilon)$ can be further reduced
    with the improvement of spectral sparsification algorithms
    \cite{SpielmanTengSpectralSparsification,SpielmanSrivastava,Koutis14:arxiv,MillerPX13b:arxiv}
    and tighter numerical analysis of our approaches.}
  constants $c_1$ and $c_2$, where $\kappa$ 
  denotes the condition number
  \footnote{As shown in \cite{SpielmanTengLinear}(Lemma 6.1),
    for an SDDM matrix $M$, $\kappa(\M{M})$ is essentially upper bounded
    by the ratio of the largest off-diagonal entry to the smallest non-zero 
    off-diagonal entry of $-\M{M}$, thus, $\kappa(M)$ is always bounded
    by ${\rm poly}(n)/\epsilon_{{\sb M}}$, where $\epsilon_{{\sb M}}$
    denotes the  machine precision.
    Note also, the running time
    of any numerical algorithm for factoring $\M{M}^{-1}$
    is $\Omega(m+ \log{(\kappa(\M{M})/\epsilon)})$ where 
    $\Omega(\log{(\kappa(\M{M})/\epsilon)})$ is the bits of precision
    required in the worst case.}
  of $\M{M}$.
Moreover, for the case when 
  $p = -1$ (i.e., the factorization problem needed for 
  Gaussian sampling), we can remove the polynomial dependency of $\epsilon$,
  so the complexity in this case is of form 
   $O(m\cdot\log^{c_1} n\cdot\log^{c_2} (\kappa/\epsilon))$.
Of equal importance, our algorithm can be
implemented to run in polylogarithmic depth on a parallel machine with
$m$ processors. 
Prior to our work, no algorithm could provide such a factorization in nearly-linear time.
Because the factorization of an SDD
  \footnote{A matrix $\MM = (m_{i,j})$ is
    {\em symmetric diagonally dominant} (SDD) if for all $i$,
    $m_{i,i} > \sum_{j \neq i} \left|m_{i,j}\right|$.}
  matrix can be reduced to the factorization of an SDDM matrix twice
  its dimensions (See Appendix \ref{sec:reduction}),
  our sampling algorithm can be directly extended to Gaussian random fields
  with SDD precision matrices. 

Our result is a step toward understanding the holy grail 
  question in parallel sampling mentioned earlier.
It provides a natural and non-trivial sufficient condition for Gaussian random fields
  that leads to parallel sampling algorithms that run in polylogarithmic parallel time and nearly-linear total work.
While we are mindful of the many classes of precision matrices that are not reducible to SDD matrices,
  we believe our approach is applicable to a wider range of matrices.
The key structural property that we use is that matrices related to the square
  of the off-diagonal entries of SDDM matrices have sparse representations.
We believe incorporating such structural properties is crucial for overcoming the
  $\Omega(mn)$-barrier for sampling from Gaussian random fields, and
  that our construction of the sparse factor chain may lead to efficient parallel
  sampling  algorithms for other popular graphical models.
By its nature of development, we hope our algorithmic advance
  will also help to strengthen the connection 
  between machine learning and spectral graph theory, 
  two of the active fields in understanding large data and networks.

\section{Background and Notation}
\label{sec:background}

In this section, we introduce the definitions and notations 
  that we will use in this paper.
We use $\rho(\M{M})$ to denote the spectral radius of the matrix $\M{M}$,
  which is the maximum absolute value of the eigenvalues of $\M{M}$.
We use $\kappa(\M{M})$ to denote its condition number corresponding to the spectral norm,
  which is the ratio between the largest and smallest singular values of $\MM$.

In our analysis, we will make extensive use of spectral approximation.
We say $\M{X} \approx_\epsilon \M{Y}$ when
\begin{align}
\exp \left(\epsilon\right) \M{X} \succcurlyeq \M{Y} \succcurlyeq  \exp \left( -\epsilon \right) \M{X},
\end{align}
where the Loewner ordering $ \M{Y} \succcurlyeq \M{X}$ means
  $\M{Y}-\M{X}$ is positive semi-definite.
 We use the following standard facts about the Loewner positive definite
 order and approximation. 

\begin{fact}\label{fact}
For positive semi-definite matrices
$\M{X}$, $\M{Y}$, $\M{W}$ and $\M{Z}$,
\begin{enumerate}
\item [a.] if $\M{Y} \approx_{\epsilon} \M{Z}$,
then $\M{X} + \M{Y} \approx_{\epsilon} \M{X} + \M{Z}$;

\item [b.] if  $\M{X} \approx_{\epsilon} \M{Y}$ and $\M{W} \approx_{\epsilon} \M{Z}$,
  then $\M{X} + \M{W} \approx_{\epsilon} \M{Y} + \M{Z} $;

\item [c.]
if $\M{X} \approx_{\epsilon_1} \M{Y}$ and
$\M{Y} \approx_{\epsilon_2} \M{Z}$,
then $\M{X} \approx_{\epsilon_1 + \epsilon_2} \M{Z}$;

\item [d.]
if $\M{X}$ and $\M{Y}$ are positive definite matrices
  such that $\M{X} \approx_{\epsilon} \M{Y}$,
  then $\M{X}^{-1} \approx_{\epsilon} \M{Y}^{-1}$;

\item [e.]
if $\M{X} \approx_{\epsilon} \M{Y}$
and $\M{V}$ is a matrix, then
$  \M{V}^{\T} \M{X} \M{V} \approx_{\epsilon} \M{V}^{\T} \M{Y} \M{V}.$
\end{enumerate}
\end{fact}

We say that $\M{C}$ is an \emph{inverse square-root factor} of 
  $\M{M}$, if $\M{C}\M{C}^{\T}=\M{M}^{-1}$.
When the context is clear, we sometime refer to $\M{C}$ as an inverse factor
  of $\M{M}$.
Note that the inverse factor of a matrix is not unique.

We will work with SDDM matrices in our algorithms,
and denote them using $\M{M}$.
SDDM matrices are positive-definite SDD matrices.
The equivalence between solving linear systems in such matrices
and graph Laplacians, weakly-SDD matrices, and M-matrices
are well known~\cite{SpielmanTengSolver,DaitchS08,Kelner,PengSpielman}.
In Appendix~\ref{sec:reduction}, we rigorously prove that
this connection carries over to inverse factors.

SDDM matrices can be split into diagonal and off-diagonal entries.
We will show in Appendix~\ref{sec:normalization} that by allowing
diagonal entries in both matrices of this splitting, we can ensure
that the diagonal matrix is $\M{I}$ without loss of generality.
\begin{restatable}{lemma}{restateIxOk}
\label{lem:ixOk}
Let $\MM=\M{D}-\M{A}$ be a SDDM matrix, where $\M{D}=\text{diag}\left(d_1,d_2,\dots,d_n \right)$ is its diagonal, and let $\kappa = \max\{ 2,\kappa\left(\M{D}-\M{A}\right) \} $. We can pick $c=\left(1 - 1/\kappa\right)/(\max_{i} d_i)$, so that all eigenvalues of $c\left( \M{D}-\M{A} \right)$ lie between $\frac{1}{2\kappa}$ and $2-\frac{1}{2\kappa}$. Moreover, we can rewrite $c\left( \M{D}-\M{A} \right)$ as a SDDM matrix $\M{I}-\M{X}$, where all entries of $\M{X}=\M{I}-c\left( \M{D}-\M{A} \right)$ are nonnegative, and $\rho\left(\M{X} \right) \leq 1-\frac{1}{2\kappa}$.
\end{restatable}

Since all of our approximation guarantees are multiplicative, we will ignore the constant
factor rescaling, and use the splitting $\M{M} = \M{I} - \M{X}$ for the rest of our analysis.

For a symmetric, positive definite matrix $\M{M}$ with spectral decomposition
\begin{align}
  \M{M} = \sum_{i = 1}^{n} \lambda_i \V{u}_i \V{u}_i^\T,
\end{align}
its $p\textsuperscript{th}$ power for any $p\in \Reals{}$ is:
\begin{align}
  \M{M}^{p} = \sum_{i = 1}^{n} \lambda_i^{p} \V{u}_i \V{u}_i^\T.
\end{align}

\section{Overview}
\label{sec:overview}
We start by briefly outlining our approach.
Our first construction of $\tilde{\CC}$ follows from the fact that
  SDDM matrices can be factorized. 
The proof of the following lemma can be found in Appendix~\ref{sec:mbfactor}.
\begin{restatable}[Combinatorial Factorization]{lemma}{restateSddFactor}
\label{lem:sddfactor}
Let $\MM$ be an SDDM matrix with $m$ non-zero entries and condition number $\kappa$, then
\begin{itemize}
\item[(1)] $\MM$ can be factored into
\begin{align}
  \M{M} = \M{B}\M{B}^{\T},
\end{align}
where $\M{B}$ is an $n \times m$ matrix with at most 
  $2$ non-zeros per column;
\item[(2)] for any approximation parameter $\epsilon>0$, 
  if $\M{Z}$ is linear operator 
  satisfying $\M{Z} \approx_\epsilon \M{M}^{-1}$, then
the matrix
\begin{align}
  \tilde{\M{C}} = \M{Z} \M{B}
\end{align}
satisfies
\begin{align}
  \tilde{\M{C}} \tilde{\M{C}}^{\T} \approx_{2\epsilon} \M{M}^{-1}.
\end{align}
Moreover, a representation of $\tilde{\M{C}}$ can be computed in work
$\OO{m\cdot \log^{c_1}n\cdot \log^{c_2}\kappa}$ and depth
$\OO{\log^{c_3}n \cdot \log \kappa}$ for some modest
  constants $c_1$, $c_2$, and $c_3$ so that for any
  $m$-dimensional vector $\V{x}$, the product $\tilde{\M{C}} \V{x}$ can
be computed in work $\OO{(m + n \cdot \log^{c}n \cdot \log^3 \kappa)\cdot 
  \log(1/\epsilon)}$ and depth $\OO{\log n \cdot 
\log \kappa \cdot \log(1/\epsilon)}$ 
for a modest constant $c$.
\end{itemize}

\end{restatable}

In particular, the parallel solver algorithm from~\cite{PengSpielman} computes
  $\M{Z} \approx_\epsilon \M{M}^{-1}$ in total work nearly-linear in $m$,
  and depth polylogarithmic in $n$, $\kappa$ and $\epsilon^{-1}$.
 We also adapt parts of \cite{PengSpielman} to obtain our main algorithm 
  with optimal  randomness complexity, which constructs
   $\tilde{\M{C}}$ that is $n \times n$.
The high level idea is to break down a difficult linear operator, $\M{M}$,
into a series of products of operators that are individually easy to evaluate.
 One candidate expression is the identity
\begin{equation}
\left(\M{I} - \M{X}\right)^{-1} = \left(\M{I} + \M{X}\right) \left(\M{I} - \M{X}^2 \right)^{-1}.
\label{eqn:factorDirect}
\end{equation}
It reduces computing $(\M{I}-\M{X})^{-1}$ to a matrix-vector multiplication and
computing $(\M{I}-\M{X}^2)^{-1}$.
 As $\Nsp{\M{X}^2} < \Nsp{\M{X}}$ when $\Nsp{\M{X}} < 1$, $\M{I} - \M{X}^2$ is closer to $\M{I}$ than $\M{I}-\M{X}$.
Formally, it can be shown that after $\log \kappa$ steps,
  where $\kappa = \kappa\left(\M{I} - \M{X}\right)$,
  the problem becomes well-conditioned.
This low iteration count coupled with the low cost of matrix-vector multiplications
then gives a low-depth algorithm.

A major issue with extending this intuition to matrices is that
squares of matrices can be dense.
The density of matrices can be reduced using sparsification,
 which leads to approximation errors.
The factorization in Equation~\ref{eqn:factorDirect} is problematic for handling such errors.
Fact \ref{fact}.e states that such spectral approximations can only be propagated
when composed with other operators symmetrically.
This issue was addressed in~\cite{PengSpielman} using an alternate decomposition
that introduces an additional term of $\M{I} + \M{X}$ on the other side.
However, this introduces an additional problem for computing a good factorization:
the overall expression is no longer a product.

 Our result is based on directly symmetrizing the identity in Equation~\ref{eqn:factorDirect}.
We use the fact that polynomials of $\M{I}$ and $\M{X}$ commute to
move a half power of the first term onto the other side:
\begin{equation}
\left(\M{I}-\M{X}\right)^{-1} = \left(\M{I}+\M{X}\right)^{1/2}
 \left(\M{I}-\M{X}^2\right)^{-1}
 \left(\M{I}+\M{X}\right)^{1/2}
\label{eqn:factorSymmetric}.
\end{equation}
This has the additional advantage that the terms can be naturally incorporated
into the factorization $\tilde{\M{C}} \tilde{\M{C}}^\T$.
Of course, this leads to the issue of evaluating half powers, or more generally
$p\textsuperscript{th}$ powers.
Here our key idea, which we will
   discuss in depth in Section~\ref{sec:maclaurin}
  and Appendix~\ref{apd:maclaurin},
is to use Maclaurin expansions.
These are low-degree polynomials which give high quality
approximations when the matrices are well-conditioned.
We show that for any well-conditioned matrix $\M{M}$ with $\M{M} \approx_{\OO{1}} \M{I}$,
  and for any constant power $|p| \le 1$, there exists a (polylogarithmic)
  degree $t$ polynomial $T_{p, t}(\cdot)$ such that
\begin{align}
T_{p, t} \left( \M{M} \right) \approx_{\epsilon} \M{M}^{p}.
\end{align}

The condition that $\rho(\M{X}) < 1$ means the expression $\M{I} + \M{X}$
is almost well conditioned: the only problematic case is when $\M{X}$ has an
eigenvalue close to $-1$.
We resolve this issue by halving the coefficient in front of $\M{X}$, leading
to the following formula for (symmetric) factorization:
\begin{equation}\label{eqn:simpleidt}
\left(\M{I}-\M{X}\right)^{-1} = \left(\M{I}+\frac{1}{2}\M{X}\right)^{1/2}
 \left(\M{I}-\frac{1}{2}\M{X}-\frac{1}{2}\M{X}^2\right)^{-1}
 \left(\M{I}+\frac{1}{2}\M{X}\right)^{1/2},
\end{equation}
which upon introduction of $T_{p, t}(\cdot)$ becomes:
\begin{equation}\label{eqn:simpleidt}
\left(\M{I}-\M{X}\right)^{-1} \approx T_{\frac{1}{2}, t} \left(\M{I}+\frac{1}{2}\M{X}\right)
 \left(\M{I}-\frac{1}{2}\M{X}-\frac{1}{2}\M{X}^2\right)^{-1}
 T_{\frac{1}{2}, t} \left(\M{I}+\frac{1}{2}\M{X}\right).
\end{equation}

Note that the two terms $T_{\frac{1}{2}, t} \left(\M{I}+\frac{1}{2}\M{X}\right)$
are symmetric.
As our goal is to approximate $\left(\M{I}-\M{X}\right)^{-1}$ with $ \tilde{\M{C}} \tilde{\M{C}}^\T$,
we can place one term into $\tilde{\M{C}}$ and the other into $\tilde{\M{C}}^\T$.
This allows us to reduce the problem to one involving
$\left(\M{I}-\frac{1}{2}\M{X}-\frac{1}{2}\M{X}^2\right)^{-1}$.
Also, as $\M{X}^2$ may be a dense matrix, we need to sparsify it during intermediate steps.
This matrix is the average of $\M{I} - \M{X}$ and $\M{I} - \M{X}^2$.
 A key observation in~\cite{PengSpielman} is that the second matrix,
  $\M{I} - \M{X}^2$, corresponds to two-step random walks on $\M{I} - \M{X}$,
  and is therefore SDDM and can be efficiently sparsified,
  which implies that we can obtain a sparsifier for the average as well.

The above numerical reduction from the factorization
  of $\left(\M{I}-\M{X}\right)^{-1}$ to the factorization of
  $\left(\M{I}-\frac{1}{2}\M{X}-\frac{1}{2}\M{X}^2\right)^{-1}$
 leads to a chain of matrices akin to the sparse inverse
chain defined in~\cite{PengSpielman}.
We call it a {\em sparse inverse factor chain},
and prove its existence in Lemma \ref{crl:chain}.
\begin{definition}[Sparse Inverse Factor Chain]
For a sequence of approximation parameters
  $\V{ \epsilon} =\left(\epsilon_0,\ldots,\epsilon_d\right)$,
  an $\V{ \epsilon}$-sparse factor chain $(\M{X}_1,\ldots,\M{X}_d)$ for $\M{M} = \M{I} - \M{X}_0$
  satisfies the following conditions:
\begin{enumerate}
\item For $i=0,1,2,\dots,d-1$, $\M{I}-\M{X}_{i+1} \approx_{\epsilon_{i}} \M{I} - \frac{1}{2} \M{X}_{i} - \frac{1}{2} \M{X}_{i}^2$;
\item $\M{I} \approx_{\epsilon_d} \M{I} - \M{X}_d$.
\end{enumerate}
\end{definition}

We summarize our first technical lemma below.
In this lemma, and the theorems that follow, we will assume without further specifying
  that our SDDM matrix $\M{M}$ is $n \times n$ and has $m$ nonzero entries.
We also also use $\kappa$ to denote $\kappa(\M{M})$, and assume that $\M{M}$
 is expressed in the form $\M{M} = \M{I}-\M{X}_0$ for a non-negative matrix $\M{X}_0$.

\begin{lemma}[Efficient Chain Construction]\label{thm:const}
Given an SDDM matrix $\M{M} = \M{I}-\M{X}_0$
  and approximation parameter $\epsilon$, there exists 
  an $\V{\epsilon}$-sparse factor chain 
  $(\M{X}_1,\M{X}_2,\dots,\M{X}_d)$ for $\M{M}$
  with $d=\OO{\log(\kappa/\epsilon)}$
  such that (1) $\sum_{i=0}^d \epsilon_i \leq \epsilon$,
  and (2) the total number of nonzero entries in 
   $(\M{X}_1,\M{X}_2,\dots,\M{X}_d)$ is upper bounded
  by $\OO{n \log^{c}n \cdot \log^{3}(\kappa/\epsilon)
  / \epsilon^{2}}$, for a constant $c$.

Moreover, we can construct such an $\V{{\epsilon}}$-sparse factor chain
 in work
  $\OO{m\cdot \log^{c_1}n \cdot \log^{c_2}(\kappa/\epsilon)/\epsilon^4}$
  and parallel depth $\OO{\log^{c_3}n\cdot \log(\kappa/\epsilon)}$,
  for some constants $c_1,c_2$ and $c_3$.
\end{lemma}

While we can use any nearly linear-time spectral sparsification algorithm 
  \cite{SpielmanTengSpectralSparsification,SpielmanSrivastava} 
  in our algorithm for Lemma \ref{thm:const}, 
  the precise exponents of the $\log{n}$ and $\log\kappa$ terms
  can be improved using recent works on combinatorial spectral 
  sparsification~\cite{Koutis14:arxiv,MillerPX13b:arxiv}.
Currently, the $c_i$s in Lemma \ref{thm:const} and the subsequent 
 Theorems \ref{thm:main} and \ref{thm:pthmain} are at most $5$.
Because of the likely possibility of further improvements in the near future,
  we only state them here as constants.

The length of the chain is logarithmic in both the condition number
  $\kappa$ and $\sum \epsilon_i$.
We will analyze the error propagation along this chain 
  in Section~\ref{sec:chain}.
Now we state our main result for computing inverse square root
  and factoring $p\textsuperscript{th}$ power in the following two theorems.

\begin{theorem}[Scalable Parallel Inverse Factorization]
\label{thm:main}
Given an SDDM matrix $\M{M} = \M{I}-\M{X}_0$
  and a precision parameter $\epsilon$,
  we can construct a representation of an $n \times n$ matrix $\tilde{\M{C}}$,
  such that $\tilde{\M{C}}\tilde{\M{C}}^{\T}\approx_{\epsilon}
  (\M{I}-\M{X}_0)^{-1}$
in work
  $\OO{m\cdot \log^{c_1}n \cdot\log^{c_2}\kappa}$
  and parallel depth $\OO{\log^{c_3}n\cdot \log \kappa}$ for some
  constants $c_1,c_2$, and $c_3$.

Moreover, for any $\V{x}\in \Reals{n}$,
  we can compute $\tilde{\M{C}} \V{x}$ in work 
$\OO{(m + n \cdot \log^{c}n\cdot \log^3\kappa)\cdot \log\log\kappa\cdot \log(1/\epsilon)}$ and depth
$\OO{\log n\cdot \log \kappa \cdot  \log\log\kappa\cdot \log(1/\epsilon)}$ for some other
constant $c$.
\end{theorem}

\begin{proof}
For computing the inverse square root factor ($p=-1$),
  we first use Lemma \ref{thm:const} with $\sum \epsilon_i = 1$
  to construct a crude sparse inverse factor chain
  with length $d = \OO{\log\kappa}$.
Then for each $i$ between $0$ and $d-1$, we use Maclaurin expansions
 from Lemma~\ref{lem:taylor} to obtain $\OO{\log\log\kappa}$-degree
   polynomials that approximate
  $(\M{I} + \Half \M{X}_i)^{1/2}$ to precision $\epsilon_i = \Theta(1/d)$.
Finally we apply a refinement procedure from Theorem \ref{thm:precond}
  to reduce error to the final target precision $\epsilon$, 
  using $O(\log(1/\epsilon))$ multiplication with the crude approximator.
The number of operations needed to compute $\tilde{\M{C}} \V{x}$
  follows the complexity of Horner's method for polynomial evaluation.
\end{proof}

 The factorization formula from Equation~\ref{eqn:simpleidt}
  can be extended to any $p\in [-1,1]$.
  This allows us to construct an $\epsilon$-approximate factor of
  $\M{M}^p$ for general $p\in [-1,1]$, with work 
  nearly-linear in $m/ \epsilon^{4}$.
It remains open if this dependency on $\epsilon$
can be reduced to $\log (1/\epsilon)$ for the general case.

\begin{theorem}[Factoring $\M{M}^p$]
\label{thm:pthmain}
For any $p \in [-1, 1]$, 
  and an SDDM matrix $\M{M} = \M{I}-\M{X}_0$, 
  we can construct a representation of an $n \times n$ matrix $\tilde{\M{C}}$
  such that $\tilde{\M{C}}\tilde{\M{C}}^{\T}\approx_{\epsilon} \M{M}^{p}$
in work
  $\OO{m\cdot \log^{c_1}n\cdot \log^{c_2}(\kappa/\epsilon)
  /\epsilon^4 }$ and depth $\OO{\log^{c_3}n\cdot
  \log(\kappa/\epsilon)}$ for some constants $c_1,c_2$, and $c_3$.

Moreover, for any $\V{x}\in \Reals{n}$,
  we can compute   $\tilde{\M{C}} \V{x}$ 
in work $\OO{(m +
  n\cdot\log^{c}n\cdot \log^3(\kappa/\epsilon)/ \epsilon^{2} )\cdot
  \log(\log(\kappa) / \epsilon) }$ and depth
  $\OO{\log n\cdot \log(\kappa/\epsilon)\cdot \log(\log(\kappa) /
  \epsilon)}$ for some other constant $c$.
\end{theorem}

\begin{proof}
Because the refinement procedure in Theorem \ref{thm:precond}
  does not apply to the general $p\textsuperscript{th}$ power case,
  we need to have a longer chain and higher order Maclaurin expansion here.
We first use Lemma \ref{thm:const} with $\sum \epsilon_i = \epsilon/2$
  to construct a matrix chain with length
  $d = \OO{\log(\kappa/\epsilon)}$.
Then for $i \le d-1$, we use Lemma \ref{lem:taylor} to obtain
  $\OO{\log(\log\kappa/\epsilon)}$-degree polynomials to approximate
  $(\M{I} + \Half \M{X}_i)^{-p/2}$
  to precision $\epsilon_i = \Theta(\epsilon/d)$.
The number of operations needed to compute $\tilde{\M{C}} \V{x}$
  follows from the length of the chain and the order of Maclaurin expansion.
\end{proof}

\section{Maclaurin Series Expansion}
\label{sec:maclaurin}

In this section, we show how to approximate 
a matrix of the form
$\left(\M{I}+\frac{1}{2}\M{X}\right)^p$.
 Because $\rho\left(\M{X}\right) < 1$, $\M{I}+\frac{1}{2}\M{X}$ is
well-conditioned.
Thus, we can approximate its
$p\textsuperscript{th}$ power to any approximation parameter $\epsilon
> 0$ using an $\OO{\log(1/\epsilon)}$-order Maclaurin expansion,
i.e., a low degree polynomial of $\M{X}$.  Moreover, since the
approximation is a polynomial of $\M{X}$, the eigenbasis is preserved.

We start with the following lemma on the residue of Maclaurin expansion,
  which we prove in Appendix~\ref{apd:maclaurin}.

\begin{restatable}[Maclaurin Series]{lemma}{restateTaylor}
\label{lem:taylorscalar}
 Fix $p \in [-1, 1]$ and $\delta \in (0, 1)$, for any error tolerance $\epsilon > 0$, there exists a $t$-degree polynomial $T_{p,t}(\cdot)$ with $t \le \frac{\log(1/(\epsilon (1-\delta)^2)}{1-\delta}$, such that for all $\lambda \in [1-\delta, 1+\delta]$,
\begin{align}
\exp(-\epsilon) \lambda^{p} \leq T_{p,t}(\lambda) \leq  \exp(\epsilon) \lambda^{p}.
\end{align}
\end{restatable}

Now we present one of our key lemmas, which provides the backbone for
bounding overall error of our sparse factor chain. We claim that
substituting $\left(\M{I}+\frac{1}{2}\M{X}\right)^p$ with its
Maclaurin expansion into Equation \ref{eqn:simpleidt} preserves the
multiplicative approximation. We prove the lemma for $p \in [-1, 1]$
which later we use for computing the $p\textsuperscript{th}$ power of
SDDM matrices. For computing the inverse square-root factor, we only
need the special case when $p=-1$.

\begin{lemma}[Factorization with Matrix Polynomials]\label{lem:taylor}
 Let $\M{X}$ be a symmetric matrix with $\rho\left(\M{X}\right) < 1$ and fix $p \in [-1, 1]$.
For any approximation parameter $\epsilon>0$, there exists a $t$-degree polynomial $T_{p,t}(\cdot)$ with $t = \OO{\log (1/\epsilon)}$, such that
\begin{align}\label{eq:idanyp}
\left(\M{I}-\M{X}\right)^{p} \approx_{\epsilon} T_{-\frac{p}{2},t}\left(\M{I}+\Half\M{X}\right)  \left( \M{I}-\frac{1}{2}\M{X}-\frac{1}{2}\M{X}^2\right)^{p} T_{-\frac{p}{2},t}\left(\M{I}+\Half\M{X}\right).
\end{align}
\end{lemma}
\begin{proof}
Take the spectral decomposition of $\M{X}$
\begin{align}
 \M{X} = \sum_i \lambda_i \V{u}_i \V{u}_i^{\T},
\end{align}
where $\lambda_i$ are the eigenvalues of $\M{X}$.
Then we can represent the left hand side of Equation \ref{eq:idanyp} as 
\begin{align}\label{eqn:taylor_proof_left}
 \begin{split}
  \left(\M{I}-\M{X}\right)^{p} &= 
  \left(\M{I}+\frac{1}{2}\M{X}\right)^{-p/2} 
  \left(\M{I}-\frac{1}{2}\M{X}-\frac{1}{2}\M{X}^2\right)^{p} 
  \left(\M{I}+\frac{1}{2}\M{X}\right)^{-p/2}
  \\
  &= \sum_i \left(1-\Half\lambda_i -\Half\lambda_i^2 \right) \left( 1+\Half \lambda_i\right)^{-p} \V{u}_i \V{u}_i^{\T},
 \end{split}
\end{align}
and the right hand side of Equation \ref{eq:idanyp} as 
\begin{align}\label{eqn:taylor_proof_right}
 \begin{split}
  T_{-\frac{p}{2},t}\left(\M{I}+\Half\M{X}\right) 
  \left(\M{I}-\frac{1}{2}\M{X}-\frac{1}{2}\M{X}^2\right)^{p}  
  T_{-\frac{p}{2},t}\left(\M{I}+\Half\M{X}\right)
  \hspace{1in} \\
  = \sum_i 
  \left(1-\Half\lambda_i -\Half\lambda_i^2 \right)
  \left( T_{-\frac{p}{2},t} \left(1+\Half \lambda_i \right) \right)^2
  \V{u}_i \V{u}_i^{\T}.
 \end{split}
\end{align}
Because $\rho(\M{X})<1$, for all $i$ we have $\left(1-\Half\lambda_i -\Half\lambda_i^2 \right)>0$.
By invoking Lemma \ref{lem:taylorscalar} with $\delta=1/2$, error tolerance $\epsilon/2$ and power $-p/2$, we can obtain the polynomial $T_{-\Half p,t}(\cdot)$ with the following property,
\begin{align}
\exp(-\epsilon) \left( 1+\Half \lambda_i\right)^{-p}
\leq
\left( T_{-\Half p,t}\left( 1+\Half \lambda_i\right) \right)^{2}
\leq
\exp(\epsilon) \left( 1+\Half \lambda_i\right)^{-p}.
\end{align}

To conclude the proof, we use this inequality inside 
  Equation \ref{eqn:taylor_proof_right} and compare Equation \ref{eqn:taylor_proof_left} and \ref{eqn:taylor_proof_right}.

\end{proof}

\section{Sparse Factor Chains: Construction and Analysis}
\label{sec:chain}

Our basic algorithm constructs 
  an approximate factor of
 $(\M{I}-\M{X})^p$, for any $p \in [-1, 1]$.
Using the factorization with Maclaurin matrix polynomials
  given by Lemma \ref{lem:taylor},
  we extend the framework of Peng-Spielman \cite{PengSpielman} for
  parallel approximation of matrix inverse
  to SDDM factorization:
Our algorithm uses the following identity
\begin{align}
\left(\M{I}-\M{X}\right)^{p} 
=
\left(\M{I}+\Half\M{X}\right)^{-p/2} 
\left(\M{I}-\Half\M{X}-\Half\M{X}^2\right)^{p} 
\left(\M{I}+\Half\M{X}\right)^{-p/2}.
\end{align}
and the fact that we can approximate $\left(\M{I}+\Half\M{X}\right)^{-p/2}$
  with its Maclaurin series expansion (Lemma \ref{lem:taylor}).
It thus numerically reduces the factorization problem of
$\left(\M{I}-\M{X}\right)^{p}$ to that of $\left(\M{I}-\Half\M{X}-\Half\M{X}^2\right)^{p}$.

The key to efficient applications of this iterative reduction
  is the spectral sparsification of $\M{I}-\Half\M{X}-\Half\M{X}^2$,
so that our matrix polynomials only use sparse matrices.
In particular, we start with $\M{I}-\M{X}_0$, and 
  at step $i$  we sparsify $\M{I}-\Half\M{X}_i-\Half\M{X}_i^2$
  to obtain its spectral approximation $\M{I}-\M{X}_{i+1}$,
  and then proceed to the next iteration until $\rho(\M{X}_d)$ is small enough,
  at which point we approximate $\M{I}-\M{X}_d$ with $\M{I}$.
 
\begin{lemma}[Sparsification of  $\M{I}-\Half \M{X} - \Half \M{X}^2$]
 Given an SDDM matrix $\M{M} = \M{I}-\M{X}$
  and a precision parameter $\epsilon$,
  we can construct an $n \times n$ non-negative matrix $\SM{X}$
  in work $\OO{m \cdot \log^{c_1}n / \epsilon^4}$
  and parallel depth $\OO{\log^{c_3}n}$, such that
  (1) $\M{I}-\SM{X}\approx_{\epsilon} \M{I}-\Half \M{X} - \Half \M{X}^2$,
  and (2) the number of non-zero entries in $\SM{X}$
  is upper bounded by $\OO{n \cdot \log^{c}n / \epsilon^{2}}$
  for some modest constants $c_1,c_3$ and $c$.
\end{lemma}
\begin{proof}
We split $\M{I}-\Half\M{X}-\Half\M{X}^2$
  into $\Half\left(\M{I}-\M{X}\right) + \Half\left(\M{I}-\M{X}^2\right)$
  and first apply the Peng-Spielman sparsification
  method (Section 6 of \cite{PengSpielman})
  to sparsify $\M{I}-\M{X}^2$.
This method involves $n$ independent sparsification tasks, 
  one per row of $\M{X}$, by sampling two-step random walks on $\M{X}$.
In other words, it sparsifies $\M{I}-\M{X}^2$
  based on $\M{X}$, without actually calculating $\M{X}^2$.
It constructs a sparse matrix $\M{X}'$   
  such that $\M{I}-\M{X}' \approx_{\epsilon/2} \M{I}-\M{X}^2$
 with  work  $\OO{m \cdot \log^{2} n / \epsilon^2 }$
  and parallel depth is $\OO{\log n}$ \cite{PengSpielman}.
Moreover, the number of nonzero entries in $\M{X}'$ upper bounded by
  $\OO{m\cdot \log n / \epsilon^2}$.

Next, we sparsify $\Half(\M{I}-\M{X})+\Half(\M{I}-\M{X}')$
  by using a spectral sparsification 
  algorithm \cite{Koutis14:arxiv,MillerPX13b:arxiv,SpielmanSrivastava,SpielmanTengSpectralSparsification}
  and compute a matrix $\SM{X}$ that satisfies
\begin{align}
  \M{I}-\SM{X}
  \approx_{\epsilon/2} \Half(\M{I}-\M{X})+\Half(\M{I}-\M{X}')
  \approx_{\epsilon/2} \M{I}-\Half \M{X} - \Half \M{X}^2
\end{align}
  in work $\OO{m \cdot \log^{c_1}n / \epsilon^4}$
  and parallel depth $\OO{\log^{c_3}n}$,
  with number of non-zero entries in $\SM{X}$ upper bounded by
  $\OO{n\cdot \log^{c}n / \epsilon^{2}}$.
A technical subtlety, unlike in \cite{PengSpielman} 
  where one can simply use approximate diagonals,
  is that we need to maintain $\M{I}$, and the non-negativeness of $\SM{X}$.
This can be achieved using the splitting technique 
  presented in Section 3.3 of Peng's thesis~\cite{Peng13:thesis}.
\end{proof}

The following basic lemma then guarantees
  spectral approximation is
  preserved under the $p\textsuperscript{th}$ power, which 
  enables the continuation of the reduction.

\begin{lemma}\label{lem:pthrootapprox}
Let $\M{A}$, $\M{B}$ be positive definite matrices with $\M{A} \approx_{\epsilon} \M{B}$, then for all $p \in [-1, 1]$, we have that 
\begin{align}
\M{A}^{p} \approx_{|p|\epsilon} \M{B}^{p}.
\end{align}
\end{lemma}
\begin{proof}
By Theorem 4.2.1 of \cite{bhatia},
  if $\M{A}\succcurlyeq \M{B}$ 
  then $\M{A}^{p} \succcurlyeq \M{B}^{p} $ for all $p \in [0, 1]$. 
Thus, if $\M{A} \approx_{\epsilon} \M{B}$, 
  then for all $p \in [0, 1]$,
\begin{align}
\exp \left(p \epsilon\right) \M{A}^p \succcurlyeq \M{B}^p \succcurlyeq \exp \left(-p\epsilon \right) \M{A}^p,
\end{align}
which implies $ \M{A}^{p} \approx_{|p|\epsilon} \M{B}^{p}$. By Fact \ref{fact}.d, the claim holds for all $p \in [-1, 0]$ as well.
\end{proof}

In the rest of the section, we 
  will bound the length of our sparse factor chain
  (Section \ref{sec:chainConvergence}),
  and show that the chain indeed 
  provides a good approximation to $(\M{I}-\M{X}_0)^p$
  (Section \ref{sec:chainError}).
In Section \ref{sec:chainRefine}, we present a refinement technique
  to further reduce the dependency of $\epsilon$ to $\log (1/\epsilon)$
  for the case when $p = \pm1$.

\subsection{Convergence of the Matrix Chain}
\label{sec:chainConvergence}

To bound the length of the sparse factor chain,
  we need to study the convergence behavior of 
  $\rho\left(\M{X}_i\right)$.
We show that a properly constructed matrix chain of
  length $\OO{\log(\kappa/\epsilon)}$ is sufficient
  to achieve $\epsilon$ precision.
 Our analysis is analogous to the one in \cite{PengSpielman}.
However, we have to deal with the first order term
  in $\M{I} - \Half \M{X} - \Half \M{X}^2$, which results a longer
  chain when $\epsilon$ is small.

We start with the following lemma, which analyzes one iteration of the chain construction.
We also take into account the approximation error 
  introduced by sparsification.

\begin{lemma}
\label{lem:eigIxx}
Let $\M{X}$ and $\SM{X}$ be nonnegative symmetric matrices such that $\M{I} \succcurlyeq \M{X}$, $\M{I} \succcurlyeq \SM{X}$, and 
\begin{align}
\M{I} - \SM{X} \approx_\epsilon \M{I} - \Half \M{X} - \Half \M{X}^2.
\end{align}
If $\rho(\M{X}) \le 1 - \lambda$, then the eigenvalues of $\SM{X}$ all lie between
$1 - \Half\left(3-\lambda\right) \exp(\epsilon)$ and $1 - \Half \left(3\lambda-\lambda^2 \right) \exp(-\epsilon)$.
\end{lemma}
\begin{proof}
If the eigenvalues of $\M{I} - \M{X}$ are between $\lambda$ and $2-\lambda$,
  the eigenvalues of $\M{I} - \Half\M{X} - \Half\M{X}^2$ will belong to $[(3\lambda-\lambda^2)/2, (3-\lambda)/2]$.
The bound stated in this lemma then follows from the fact
  that the spectral sparsification step preserves the eigenvalues of
  $\M{I} - \Half\M{X} - \Half\M{X}^2$ to 
  within a factor of $\exp(\pm \epsilon)$.
  \end{proof}

We next prove that our sparse matrix chain achieves overall 
  precision $\epsilon$.

\begin{lemma}[Short Factorization Chain]\label{crl:chain}
Let $\M{M} = \M{I} - \M{X}_0$ be 
  an SDDM matrix with condition number $\kappa$.
For any approximation parameter $\epsilon$, there exists an $(\epsilon_0, \ldots, \epsilon_d)$-sparse factor chain of $\M{M}$ with 
  $d = \log_{9/8}(\kappa/\epsilon)$ and  $\sum_{j=0}^d \epsilon_j \leq \epsilon$, satisfying
\begin{align}\label{eqn:chaincondition}
 \begin{split}
  \M{I}-\M{X}_{i+1} & \approx_{\epsilon_{i}} \M{I} - \Half \M{X}_{i} - \Half \M{X}_{i}^2
  	\qquad \forall~0 \le i \le d-1,  \\
  \M{I}-\M{X}_{d} & \approx_{\epsilon_{d}} \M{I}.
 \end{split}
\end{align}
\end{lemma}

\begin{proof}
Without loss of generality, we assume $\epsilon \leq \frac{1}{10}$.
For $0 \le i \le d-1$, we set $\epsilon_i = \frac{\epsilon}{8d}$.
We will show that our chain satisfies the condition
 given by Equation (\ref{eqn:chaincondition}) with $\epsilon_d \leq (7\epsilon)/8$, and thus
  $\sum_{i=0}^d \epsilon_i\leq \epsilon.$

Let $\lambda_i = 1- \rho(\M{X}_i)$.
We split the error analysis to two parts,
  one for $\lambda_i \le \Half$,  and one for $\lambda_i \ge \Half$.

In the first part, 
  because $\epsilon_i \le 1/10$, 
\begin{align}
\begin{split}
1 - \Half \left(3\lambda_i-\lambda_i^2 \right) \exp(-\epsilon_i) &<
1- \frac{9}{8} \lambda_i \quad \mbox{and} \\
1 - \Half\left(3-\lambda_i\right) \exp(\epsilon_i) &>
-1+ \frac{9}{8} \lambda_i
\end{split}
\end{align}
which, by Lemma \ref{lem:eigIxx}, implies that 
  $\lambda_{i+1} \ge \frac{9}{8}\lambda_i$.
Because initially, $\lambda_0 \geq 1/\kappa$,
  after $d_1 = \log_{9/8}\kappa$ iterations,
  it must be the case that $\lambda_{d_1} \ge \Half$.
 
Now we enter the second part.
Note that $\epsilon_i = \frac{\epsilon}{8d} \le \frac{\epsilon}{6}$.
Thus, when $\Half \le \lambda_i \le 1 - \frac{5}{6} \epsilon$,
  we have $\epsilon_i \le \frac{1-\lambda_i}{5}$ and
  therefore
\begin{align}
\begin{split}
1 - \Half \left(3\lambda_i-\lambda_i^2 \right) \exp(-\epsilon_i) & <
\frac{8}{9}(1-\lambda_i) \quad \mbox{and} \\
1 - \Half\left(3-\lambda_i\right) \exp(\epsilon_i) &>
\frac{8}{9}(-1+\lambda_i)
\end{split}
\end{align}
which,  by Lemma \ref{lem:eigIxx}, 
 implies that $1-\lambda_{i+1} \le \frac{8}{9}(1-\lambda_i)$.
Because $1-\lambda_{d_1} \le \Half$,
  after another $d_2 = \log_{9/8}(1/\epsilon)$ iterations,
  we get $\rho\left(\M{X}_{d_1+d_2}\right) \le \frac{5}{6} \epsilon$.
Because $\exp(-\frac{7}{8} \epsilon) < 1 - \frac{5}{6}\epsilon$
  when $\epsilon \le \frac{1}{10}$,
we conclude that $\left(\M{I} - \M{X}_{d_1+d_2}\right) \approx_{\epsilon_d} \M{I}$
  with $\epsilon_d \le \frac{7}{8}\epsilon$.
\end{proof}

\subsection{Precision of the Factor Chain}
\label{sec:chainError}

We now bound the total error of the matrix factor that our algorithm
  constructs.
We start with an error-analysis of a single reduction step 
  in the algorithm.
\begin{lemma}\label{lem:errOne}
Fix $p \in [-1, 1]$, given an $\V{\epsilon}$-sparse factor chain,
  for all $0 \le i \le d-1$, there exists degree $t_i$ polynomials,
  with $t_i=\OO{\log(1/\epsilon_i)}$, such that
\begin{align}
(\M{I}-\M{X}_{i})^{p} 
\approx_{2 \epsilon_{i}}  
T_{-\Half p,t_i}\left(\M{I}+\Half\M{X}_i\right)
\left(\M{I} -\M{X}_{i+1}\right)^{p}
T_{-\Half p,t_i}\left(\M{I}+\Half\M{X}_i\right).
\end{align}
\end{lemma}
\begin{proof}
\begin{align}
\begin{split}
\left(\M{I}-\M{X}_{i}\right)^{p} 
&\approx_{\epsilon_{i}} 
T_{-\Half p,t_i}\left(\M{I}+\Half\M{X}_i\right)
\left( \M{I}-\Half\M{X}_{i}-\Half \M{X}^2_i \right)^{p}
T_{-\Half p,t_i}\left(\M{I}+\Half\M{X}_i\right) \\
&\approx_{\epsilon_{i}} 
T_{-\Half p,t_i}\left(\M{I}+\Half\M{X}_i\right)
\left(\M{I} -\M{X}_{i+1}\right)^{p}
T_{-\Half p,t_i}\left(\M{I}+\Half\M{X}_i\right).
\end{split}
\end{align}

The first approximation follows from  Lemma \ref{lem:taylor}.
The second approximation follows from 
Lemma \ref{lem:pthrootapprox}, Fact \ref{fact}.e, and 
$
\M{I}-\M{X}_{i+1} \approx_{\epsilon_i} \M{I}-\Half\M{X}_{i}-\Half \M{X}^2_i
$. 
\end{proof}

The next lemma bounds the total error of our construction.

\begin{lemma}\label{lem:error}
Define $\M{Z}_{p,i}$ as
\begin{align}
  \M{Z}_{p,i} =
  \begin{cases}
    \M{I} & \text{if $i = d$}\\
    \left(T_{-\Half p,t_i}\left(\M{I}+\Half\M{X}_i\right)\right) \M{Z}_{p,i+1} & \text{if }0 \le i \le d-1.
  \end{cases}
\end{align}
Then the following statement is true for all $0 \le i \le d$
\begin{align}
  \M{Z}_{p,i} \M{Z}^{\T}_{p,i} \approx_{2\sum_{j=i}^{d} \epsilon_j} \left(\M{I}-\M{X}_i\right)^{p}.
\end{align}
\end{lemma}
\begin{proof}
We prove this lemma by reverse induction. Combine $\M{I} \approx_{\epsilon_d} \M{I} - \M{X}_d$ and Lemma \ref{lem:pthrootapprox},
  we know that $\M{Z}_{p,d} \M{Z}^{\T}_{p,d} = \M{I}^{p} \approx_{\epsilon_d} \left(\M{I}-\M{X}_d\right)^{p}$,
  so the claim is true for $i=d$.
Now assume the claim is true for $i=k+1$, then
\begin{align}
\begin{split}
\M{Z}_{p,k} \M{Z}^{\T}_{p,k}
= &
T_{-\Half p,t_k}\left(\M{I}+\Half\M{X}_k\right)
\M{Z}_{p,k+1} \M{Z}^{\T}_{p,k+1}
T_{-\Half p,t_k}\left(\M{I}+\Half\M{X}_k\right) \\
{\approx}&{}_{2 \sum_{j=k+1}^d \epsilon_j} 
T_{-\Half p,t_k}\left(\M{I}+\Half\M{X}_k\right)
\left(\M{I}-\M{X}_{k+1}\right)^{p}
T_{-\Half p,t_k}\left(\M{I}+\Half\M{X}_k\right) \\
{\approx}&{}_{2 \epsilon_k} \left(\M{I}-\M{X}_{k}\right)^{p}.
\end{split}
\end{align}
The last approximation follows from
  Lemma \ref{lem:errOne}.
\end{proof}

If we define $\tilde{\M{C}}$ to be $\M{Z}_{-\Half,0}$,
  this lemma implies $\tilde{\M{C}} \tilde{\M{C}}^\T \approx_{\sum_{j=0}^d \epsilon_j} (\M{I}-\M{X}_0)^{-1}$.

\subsection{Iterative Refinement for Inverse Square-Root Factor}
\label{sec:chainRefine}

In this section, 
  we provide a highly 
  accurate factor 
  for the case when $p=-1$.
We use a refinement approach inspired by the precondition technique for square factors
described in Section 2.2 of~\cite{chow14precon}.

\begin{theorem}\label{thm:precond}
For any symmetric positive definite matrix $\MM$, any matrix $\M{C}$
  such that $\MM^{-1} \approx_{\OO{1}} \M{Z} \M{Z}^\T$,
  and any error tolerance $\epsilon > 0$, there exists
  a linear operator $\tilde{\M{C}}$ which is an $\OO{\log(1/\epsilon)}$-degree
   polynomial of $\MM$ and $\M{Z}$, such that
  $\tilde{\M{C}} \tilde{\M{C}}^\T \approx_\epsilon \MM^{-1}$.
\end{theorem}

\begin{proof}

Starting from $\MM^{-1} \approx_{\OO{1}} \M{Z} \M{Z}^\T$, we know that
  $\M{Z}^\T \MM \M{Z} \approx_{\OO{1}} \M{I}$.
This implies that $\kappa(\M{Z}^\T \MM \M{Z}) = \OO{1}$, which we can scale $\M{Z}^\T \MM \M{Z}$ so that its eigenvalues lie in $[1-\delta,1+\delta]$ for some constant $0<\delta<1$.

By applying Lemma \ref{lem:taylorscalar} on its eigenvalues, there is an $\OO{\log(1/\epsilon)}$ degree
  polynomial $T_{-\Half, t}(\cdot)$ that approximates the inverse square root of
  $\M{Z}^\T \MM \M{Z}$, i.e.
\begin{align}
  \left(T_{-\Half, \OO{\log(1/\epsilon)}} \left(\M{Z}^\T \MM \M{Z}\right)\right)^2 \approx_\epsilon \left(\M{Z}^\T \MM \M{Z}\right)^{-1}.
\end{align}
Here, we can rescale $\M{Z}^\T \MM \M{Z}$ back inside $T_{-\Half, t}(\cdot)$, which does not affect the multiplicative error. Then by Fact \ref{fact}.e, we have
\begin{align}
  \M{Z} \left(T_{-\Half, \OO{\log(1/\epsilon)}}\left(\M{Z}^\T \MM \M{Z}\right)\right)^2 \M{Z}^\T \approx_\epsilon \M{Z} \left(\M{Z}^\T \MM \M{Z}\right)^{-1} \M{Z}^\T = \MM^{-1}.
\end{align}
So if we define the linear operator
  $\tilde{\M{C}} = \M{Z} \left(T_{-\Half, \OO{\log(1/\epsilon)}}\left(\M{Z}^\T \MM \M{Z}\right) \right)$,
  $\tilde{\M{C}}$ satisfies the claimed properties.
\end{proof}
Note that this refinement procedure is specific to $p = \pm 1$.
It remains open if it can be extended to general $p\in [-1,1]$.

\bibliographystyle{alpha}
\bibliography{isr,parallelSampling}
\begin{appendix}
\section{Gremban's Reduction for Inverse Square Root Factor}
\label{sec:reduction}

We first note that Gremban's reduction \cite{Gremban} can be extended from solving linear systems to computing the inverse square-root factor. The problem of finding a inverse factor of a SDD matrix, can be reduced to finding a inverse factor of a SDDM matrix that is twice as large.
\begin{lemma}
\label{lem:sddm}
Let $\M{\Lambda}=\M{D}+\M{A}_{n}+\M{A}_{p}$ be an $n\times n$ SDD matrix,
  where $\M{D}$ is the diagonal of $\M{\Lambda}$,
  $\M{A}_{p}$ contains all the positive off-diagonal entries of $\M{\Lambda}$,
  and $\M{A}_{n}$ contains all the negative off-diagonal entries.

Consider the following SDDM matrix $\M{S}$, and an inverse factor $\tilde{\M{C}}$ of $\M{S}$, such that
\begin{align}
\M{S}=\left[
\begin{array}{cc}
\M{D}+\M{A}_{n} & -\M{A}_{p}\\
-\M{A}_{p} & \M{D}+\M{A}_{n}
\end{array}
\right] \text{ and }
\tilde{\M{C}} \tilde{\M{C}}^\T = \M{S}.
\end{align}

Let $\M{I}_{n}$ be the $n\times n$ identity matrix.
The  matrix
\begin{align}
\M{C} =\frac{1}{\sqrt{2}}\left[\begin{array}{cc}
\M{I}_{n} &
-\M{I}_{n} \\
\end{array}\right]\tilde{\M{C}}
\end{align}
 is an inverse square-root factor of $\M{\Lambda}$.

\end{lemma}

\begin{proof}
We can prove that 
\begin{align}
\M{\Lambda}^{-1}=\frac{1}{2}\left[\begin{array}{cc}
\M{I}_{n} & -\M{I}_{n}\end{array}\right]\M{S}^{-1}\left[\begin{array}{c}
\M{I}_{n}\\
-\M{I}_{n}
\end{array}\right]
\end{align}

thus we have $\M{C}\M{C}^{\T}=\M{\Lambda}^{-1}$.
\end{proof}
Note that if $\tilde{\M{C}}$ is $2n \times 2n$, $\M{C}$ will be an $n\times 2n$ matrix.
Given the non-square inverse factor $\M{C}$, we can draw a sample from multivariate Gaussian distribution
  with covariance $\M{\Lambda}^{-1}$ as follows.
First draw $2n$ i.i.d. standard Gaussian random variables
  $\V{x}=(x_{1},x_{2},\dots,x_{2n})^{\T}$,
  then $\V{y}=\M{C}\V{x}$ is a multivariate Gaussian random vector with covariance
  $\M{\Lambda}^{-1}$.

\section{Linear Normalization of SDDM Matrices}
\label{sec:normalization}
For a SDDM matrix $\M{D}-\M{A}$, one way to normalize it is to multiply $\M{D}^{-1/2}$ on both sides,
\begin{align}
\M{D}-\M{A} = \M{D}^{1/2} \left( \M{I} - \M{D}^{-1/2}\M{A}\M{D}^{-1/2} \right) \M{D}^{1/2}.
\end{align}
Because $\rho\left( \M{D}^{-1/2}\M{A}\M{D}^{-1/2}\right) < 1$, this enables us to approximate $\left( \M{I} - \M{D}^{-1/2}\M{A}\M{D}^{-1/2} \right)^p$. It is not clear how to undo the normalization and obtain $\left(\M{D}-\M{A}\right)^p$ from it, due to the fact that $\M{D}$ and $\M{A}$ do not commute. Instead, we consider another normalization as in Lemma \ref{lem:ixOk} and provide its proof.

\restateIxOk*

\begin{proof}
Let $j=\arg\max_i d_i$, by definition $c = \left(1 - 1/\kappa\right) / d_j$. 

Let $\left[\lambda_{\min},\lambda_{\max}\right]$ denote the range of the eigenvalues of $\M{D}-\M{A}$.
By Gershgorin circle theorem and the fact that $\M{D}-\M{A}$ is diagonally dominant, we have that $\lambda_{\max} \leq 2 d_j$.
Also $d_j = \V{e}_j^{\T} \left( \M{D}-\M{A} \right) \V{e}_j$, so $\lambda_{\max} \geq d_j$, and therefore $\lambda_{\min} \geq \lambda_{\max} / \kappa \geq d_j / \kappa$. 

The eigenvalues of $c\left( \M{D}-\M{A} \right)$ lie in the interval $\left[c \lambda_{\min},c \lambda_{\max}\right]$, and can be bounded as
\begin{align}
\frac{1}{2\kappa} \leq \left(1-\frac{1}{\kappa}\right)\frac{1}{d_j} \left(\frac{d_j}{\kappa}\right) \leq c\lambda_{\min}\leq c\lambda_{\max}\leq  \left(1-\frac{1}{\kappa}\right)\frac{1}{d_j} \left(2 d_j\right) \leq 2-\frac{1}{2\kappa}.
\end{align}
To see that $\M{X} = \left(\M{I}-c \M{D}\right)+c\M{A}$ is a nonnegative matrix, note that $\M{A}$ is nonnegative and $\M{I}$ is entrywise larger than $c\M{D}$. 
\end{proof}

\section{Inverse Factor by SDDM Factorization}
\label{sec:mbfactor}
We restate and prove Lemma \ref{lem:sddfactor}.

\restateSddFactor*

\begin{proof}
Let $\V{e}_i$ be the $i\textsuperscript{th}$ standard basis, then $\M{M}$ can be represented by
\begin{align}
\M{M} =\frac{1}{2} \sum_{1 \le i < j \le n} \left| \Mij{\M{M}}{i}{j} \right|
 \left(\V{e}_i -\V{e}_j \right)
\left(\V{e}_i - \V{e}_j \right)^{\T}
+ \sum_{i=1}^{n} a_i \V{e}_i \V{e}_i^{\T},
\end{align}
where $a_i =  \Mij{\M{M}}{i}{j} -\sum_{j \neq i}\left| \Mij{\M{M}}{i}{j} \right| \geq 0$. Indeed, we can rewrite $\M{M}$ as $\sum_{i=1}^m \V{y}_i\V{y}_i^{\T}$ this way, and by setting $\M{B}=\left(\V{y}_1,\V{y}_2,\dots,\V{y}_m \right)$ concludes the proof for statement (1).

For statement (2),
\begin{align}
\begin{split}
\tilde{\M{C}}\tilde{\M{C}}^\T &= \M{Z} \M{B}\M{B}^\T \M{Z}^\T \\
  &= \M{Z} \MM \M{Z} \\
  &\approx_{\epsilon} \M{Z} \M{Z}^{-1} \M{Z} = \M{Z}.
\end{split}
\end{align}
The third line is true because of Fact \ref{fact}.d and Fact \ref{fact}.e.
From $\tilde{\M{C}}\tilde{\M{C}}^\T \approx_{\epsilon} \M{Z}$, by Fact \ref{fact}.c,
  we have $\tilde{\M{C}}\tilde{\M{C}}^\T \approx_{2\epsilon} \M{M}^{-1}$.

The total work and depth of constructing the operator $\M{Z}\M{B}$
  is dominated by the work and depth of constructing the approximate inverse $\M{Z}$,
  which we refer to the analysis in \cite{PengSpielman}.
The same statement also holds for the work and depth of matrix-vector multiplication of $\M{Z}\M{B}$.
\end{proof}

\section{Residue of Maclaurin Expansion}
\label{apd:maclaurin}
In this section we prove Lemma \ref{lem:taylorscalar}. We start with the following proposition.
\begin{proposition} \label{prp:scaleTaylor}
For $\left|x\right|\leq \delta <1$, we have 
\begin{align}
\sup_{\left|x\right|\leq \delta} \left|f_p\left(x\right)-g_{p,t}\left(x\right)\right| \leq \frac{\delta^{t+1}}{1-\delta}.
\end{align}
where $g_{p,t}\left(x\right)$ is the $t\textsuperscript{th}$ order Maclaurin series expansion of $f_p(x)=\left(1-x \right)^{p}$.
\end{proposition}
\begin{proof}
We can represent $g_{p,t}\left(x\right)$ as $\sum_{i=0}^{\infty} a_i x^i$.

We prove $\left|a_i\right| \leq 1$ for all $i$ by induction, because $\left|a_1\right| = \left|p\right| \leq 1$, and
\begin{align}
\left|a_{k+1}\right| = \left|\frac{p-k}{k+1}\right| \left|a_{k}\right| \leq \left|a_{k}\right|.
\end{align}
Therefore,
\begin{align}
\left|f_p\left(x\right)-g_{p,t}\left(x\right)\right|=\left|\sum_{i=t+1}^{\infty} a_i x^i \right| \leq \sum_{i=t+1}^{\infty} \delta^{i} = \frac{\delta^{t+1}}{1-\delta}. 
\end{align}
\end{proof}

\restateTaylor*

\begin{proof}
Let $g_{p,t}\left( x \right)$ be the $t\textsuperscript{th}$ order Maclaurin series expansion of $f\left(x\right)=\left(1-x\right)^{p}$. Because $\left|1-\lambda\right| \le \delta$, by Proposition \ref{prp:scaleTaylor},
\begin{align}
\left| g_{p,t}\left(1-\lambda\right) - \lambda^p \right| \leq \frac{\delta^{t+1}}{1-\delta}.
\end{align}

We define $T_{p,t}\left(x\right)=g_{p,t}\left(1-x \right)$. Also $\lambda^p \geq 1-\delta$, because $p \in [-1, 1]$ and $\lambda \in [1-\delta,1+\delta]$. So we have
\begin{align}
\left| T_{p,t}\left(\lambda\right) - \lambda^p \right| \leq \frac{\delta^{t+1}}{(1-\delta)^2} \lambda^p.
\end{align}

Therefore, when $t = (1-\delta)^{-1}  \log(1/(\epsilon(1-\delta)^2)))$, we get $\frac{\delta^{t+1}}{(1-\delta)^2} < \epsilon$ and 
\begin{align}
  \left| T_{p,t}\left(\lambda\right) - \lambda^p \right| \leq  \epsilon \lambda^p,
\end{align}

which implies
\begin{align}
\exp(-\epsilon) \lambda^{p} \leq T_{p,t}\left(\lambda\right) \leq  \exp(\epsilon) \lambda^{p}. 
\end{align}
\end{proof}
\end{appendix}

\end{document}